\documentclass[letterpaper,11pt,final]{article}
\usepackage{fullpage}
\usepackage{amsthm,amssymb,amsmath}
\def\ack{\section*{Acknowledgements}
  \addtocontents{toc}{\protect\vspace{6pt}}
  \addcontentsline{toc}{section}{Acknowledgements}
}
\def\abs#1{{|\,#1\,|}}
\def\set#1{{\{\,#1\,\}}}
\def\factor#1#2#3{{#1[#2\,..\,#3]}}
\newtheorem{lemma}{Lemma}
\newtheorem{theorem}[lemma]{Theorem}
\newtheorem{proposition}[lemma]{Proposition}
\begin{document}
\title{Multi-Shift de~Bruijn Sequence} 
\author{Zhi Xu\footnote{Part of the work was done during the author's stay at University of Waterloo}\\
The University of Western Ontario,\\
Department of Computer Science,\\
London, Ontario, Canada N6A 5B7\\
{\tt zhi\char`\_xu@csd.uwo.ca}}
\date{}
\maketitle

\begin{abstract}
A (non-circular) de~Bruijn sequence $w$ of order $n$ is a word such
that every word of length $n$ appears exactly once in $w$ as a
factor. In this paper, we generalize the concept to a multi-shift
setting: a $m$-shift de~Bruijn sequence of order $n$ is a word such
that every word of length $n$ appears exactly once in $w$ as a
factor that starts at an index $im+1$ for some integer $i\geq0$. We
show the number of the $m$-shift de~Bruijn sequences of order $n$ is
$a^n!a^{(m-n)(a^n-1)}$ for $1\leq n\leq m$ and is $(a^m!)^{a^{n-m}}$
for $1\leq m\leq n$, where $a$ is the size of the alphabet. We
provide two algorithms for generating a multi-shift de~Bruijn
sequence. The multi-shift de~Bruijn sequence is important in solving
the Frobenius problem in a free monoid.
\end{abstract}

\section{Introduction}
If a word $w$ can be written as $w=xyz$, then words $x$, $y$, and
$z$ are called the prefix, factor, and suffix of $w$, respectively.
A word $w$ over $\Sigma$ is called a de~Bruijn sequence of order
$n$, if each word in $\Sigma^n$ appears exactly once in $w$ as a
factor. For example, $\tt00110$ is a binary de~Bruijn sequence of
order $2$ since each binary word of length two appears in it exactly
once as a factor: $\tt00110=(00)110=0(01)10=00(11)0=001(10)$.
The de~Bruijn sequence can be understood by the following game.
Suppose there are infinite supplies of balls, each of which is
labeled by a letter in $\Sigma$, and there is a glass pipe that can
hold balls in a vertical line. On the top of that pipe is an
opening, through which one can drop balls into that pipe, and on the
bottom is a trap-door, which can support the weight of at most $n$
balls. When there are more than $n$ balls in the pipe, the trap-door
opens and those balls at the bottom drop off until only $n$ balls
remain. If we put balls as numbered as in a de~Bruijn sequence on
the alphabet $\Sigma$ of order $n$, then every $n$ ball sequence
will appear exactly once in the pipe. It is easy to see that a
de-Bruijn sequence of order $n$, if exists, is of length
$\abs{\Sigma}^n+n-1$ and its suffix of length $n-1$ is identical to
its prefix of length $n-1$. So, sometimes a de-Bruijn sequence is
written in a circular form by omitting the last $n-1$ letters, which
can be viewed as the equivalence class of words under the conjugate
relation.

The de~Bruijn sequence is also called the de~Bruijn-Good sequence,
named after de~Bruijn~\cite{deBruijn:1946} and Good~\cite{Good:1946}
who independently studied the existence of such words over binary
alphabet; the former also provided a formula $2^{2^n}$ for the total
number of those words of order $n$. The study of the de~Bruijn
sequence, however, dates back at least to 1894, when
Flye~Sainte-Marie~\cite{FlyeSainteMarie:1894} studied the words and
provided the same formula $2^{2^n}$. For an arbitrary alphabet
$\Sigma$, van~Aardenne-Ehrenfest and
de~Bruijn~\cite{vanAardenne-Ehrenfest&deBruijn:1951} provided the
formula $(\abs{\Sigma}!)^{\abs{\Sigma}^n}$ for the total number of
de~Bruijn sequences of order $n$. Besides the total number of
de~Bruijn sequences, another interesting topic is how to generate a
de~Bruijn sequence (arbitrary one, lexicographically least one,
lexicographically largest one). For generating de~Bruijn sequences,
see the surveys~\cite{Fredricksen:1982,Ralston:1982}. The de~Bruijn
sequence is some times called the full
cycle~\cite{Fredricksen:1982}, and has connections to the following
concepts: feedback shift registers~\cite{Golomb:1967}, normal
words~\cite{Good:1946}, generating random binary
sequences~\cite{Knuth1969}, primitive polynomials over a Galois
field~\cite{Rees:1946}, Lyndon words and
necklaces~\cite{Fredricksen&Kessler1977}, Euler tours and spanning
trees~\cite{vanAardenne-Ehrenfest&deBruijn:1951}.

In this paper, we consider a generalization of the de~Bruijn
sequence. To understand the concept, let us return to the glass pipe
game presented at the beginning. Now the trap-door can support more
weight. When there are $n+m$ or more balls in the pipe, the
trap-door opens and the balls drop off until there are only $n$
balls in the pipe. Is there an arrangement of putting the balls such
that every $n$ ball sequence appears exactly once in the pipe? The
answer is ``Yes'' for arbitrary positive integers $m,n$. The
solution represents a multi-shift de~Bruijn sequence. We will
discuss the existence of the multi-shift de~Bruijn sequence, the
total number of multi-shift de~Bruijn sequences, generating a
multi-shift de~Bruijn sequence, and the application of the
multi-shift de~Bruijn sequence in the Frobenius problem in a free
monoid.

\section{Multi-Shift Generalization of the de~Bruijn Sequence}
Let $\Sigma\subseteq\set{\tt0,1,\ldots}$ be the \emph{alphabet} and
let $w=a_1a_2\cdots a_n$ be a word over $\Sigma$. The \emph{length}
of $w$ is denoted by $\abs{w}=n$ and the \emph{factor} $a_i\cdots
a_j$ of $w$ is denoted by $\factor{w}{i}{j}$. If
$u=\factor{w}{im+1}{im+n}$ for some non-negative integer $i$, we say
factor $u$ appears in $w$ at \emph{a modulo $m$ position}. The set
of all words of length $n$ is denoted by $\Sigma^n$ and the set of
all finite words is denoted by
$\Sigma^*=\set{\epsilon}\cup\Sigma\cup\Sigma^2\cdots$, where
$\epsilon$ is the \emph{empty word}. The concatenation of two words
$u,v$ is denoted by $u\cdot v$, or simply $uv$.

A word $w$ over $\Sigma$ is called a \emph{multi-shift de~Bruijn
sequence} of shift $m$ and order $n$, if each word in $\Sigma^n$
appears exactly once in $w$ as a factor at a modulo $m$ position.
For example, one of the $2$-shift de~Bruijn sequence of order $3$ is
$\tt00010011100110110$, which can be verified as follows:
\begin{align*}
  \tt00010011100110110=(000)10011100110110=00(010)011100110110&\\
  \tt=0001(001)1100110110=000100(111)00110110=00010011(100)110110&\\
  \tt=0001001110(011)0110=000100111001(101)10=00010011100110(110)&.
\end{align*}
The multi-shift de~Bruijn sequence generalizes the de~Bruijn
sequence in the sense de~Bruijn sequences are exactly $1$-shift
de~Bruijn sequences of the same order. It is easy to see that the
length of each $m$-shift de~Bruijn sequence of order $n$, if exists,
is equal to $m\abs{\Sigma}^n+(n-m)$. By the definition of
multi-shift de~Bruijn sequence, the following proposition holds.

\begin{proposition}\label{lemma:circular}
Let $w$ be one $m$-shift de~Bruijn sequence $w$ of order $n$, $n>m$.
Then the suffix of length $n-m$ of $w$ is identical to the prefix of
length $n-m$ of $w$.
\end{proposition}
\begin{proof}
Let $w$ be one $m$-shift de~Bruijn sequence $w$ of order $n$ over
$\Sigma$ and let $a=\abs{\Sigma}$. Write $n=km+r$ such that $0<r\leq
m$. If $k=1$, then we compare the set of all factors
$\factor{w}{(i+1)m+1}{(i+1)m+r}$ and the set of all factors
$\factor{w}{im+1}{im+r}$ for $0\leq i\leq a^n-1$. The former covers
factors $\factor{u}{m+1}{m+r}$ and the latter covers factors
$\factor{u}{1}{r}$ for every $u\in\Sigma^n$. Since the two are
identical, we have $\factor{w}{a^nm+1}{a^nm+r}=\factor{w}{1}{r}$.
Now we assume $k\geq2$. Consider the set of all factors
$\factor{w}{(i+j+1)m+1}{(i+j+2)m}$ and the set of all factors
$\factor{w}{(i+j)m+1}{(i+j+1)m}$ for $0\leq i\leq a^n-1$ and $0\leq
j<k$. By the same argument, we have
$\factor{w}{(a^n+j)m+1}{(a^n+j+1)m}=\factor{w}{jm+1}{(j+1)m}$ for
$0\leq j<k$. Finally, comparing the set of all
$\factor{w}{(i+k)m+1}{(i+k)m+r}$ and the set of all
$\factor{w}{(i+k-1)m+1}{(i+k-1)m+r}$ for $0\leq i\leq a^n-1$, we
have the equality
$\factor{w}{(a^n+k-1)m+1}{(a^n+k-1)m+r}=\factor{w}{(k-1)m+1}{(k-1)m+r}$.
Therefore, we have the equality
$\factor{w}{ma^n+1}{ma^n+n-m}=\factor{w}{a^nm+1}{(a^n+k-1)m+r}=\factor{w}{1}{(k-1)m+r}=\factor{w}{1}{n-m}$.
\end{proof}

From Proposition~\ref{lemma:circular}, we know that when $n>m$,
every multi-shift de~Bruijn sequence can be written as a circular
word and the discussion on multi-shift de~Bruijn sequences of the
two different forms are equivalent. In this paper, we discuss the
multi-shift de~Bruijn sequence in the form of ordinary words.

A \emph{(non-strict) directed graph}, or \emph{digraph} for short,
is a triple $G=(V,A,\psi)$ consisting of a set $V$ of
\emph{vertices}, a set $A$ of \emph{arcs}, and an \emph{incidence
function} $\psi:A\to V\times V$. Here we do not take the convention
$A\subseteq V\times V$, since we allow a digraph contains self-loops
and multiple arcs regarding the same pair of vertices. When
$\psi(a)=(u,v)$, we say the arc $a$ joins $u$ to $v$, where vertex
$u=tail(a)$ and vertex $v=head(a)$ are called \emph{tail} and
\emph{head}, respectively. The indegree $\delta^-(v)$ (outdegree
$\delta^+(v)$, respectively) of a vertex $v$ is the number of arcs
with $v$ being the head (the tail, respectively). A \emph{walk} in
$G$ is a sequence $a_1,a_2,\ldots,a_k$ such that
$head(a_i)=tail(a_{i+1})$ for each $1\leq i<k$. The walk is
\emph{closed}, if $head(a_k)=tail(a_0)$. Two closed walks are
regarded as identical if one is the circular shift of the other. An
\emph{Euler tour} is a closed walk that traverses each arc exactly
once. A \emph{Hamilton cycle} is a closed walk that traverses each
vertex exactly once. An \emph{(spanning) arborescence} is a digraph
with a particular vertex, called the \emph{root}, such that it
contains every vertices of $G$, its number of arcs is exactly one
less than the number of vertices, and there is exactly one walk from
the root to any other vertex. We denote the total number of Euler
tours, Hamilton cycles, and arborescence of $G$ by $\abs{G}_E$,
$\abs{G}_H$, and $\abs{G}_A$, respectively.

An \emph{(undirected) graph} is defined as a digraph such that for
any pair of vertices $v_1,v_2$, there is an arc $a$,
$\psi(a)=(v_1,v_2)$, if and only if there is a corresponding arc
$a'$, $\psi(a')=(v_2,v_1)$. In this case, we write
$\delta^-(v)=\delta^+(v)=\delta(v)$ and a spanning arborescence is
just a \emph{spanning tree}.

The arc-graph $G^*$ of $G=(V,A,\psi)$ is defined as $(A,C,\varphi)$
such that for every pair of arcs $a_1,a_2\in A$,
$head(a_1)=tail(a_2)$, there is an arc $c\in C$,
$\varphi(c)=(a_1,a_2)$ and those arcs are the only arcs in $C$.
Euler tours exist in a graph $G$ if and only if Hamilton cycles
exist in the arc-graph $G^*$.

We define the word graph $G(m,n)$ by
$(\Sigma^{n},\Sigma^{n+m},\psi)$, where $\psi(w)=(u,v)$ for
$u=\factor{w}{1}{n},v=\factor{w}{m+1}{m+n}$. Then by definition, the
following lemmas are straightforward.
\begin{lemma}\label{lemma:arcgraph}
The digraph $G(m,n)^*$ is the digraph $G(m,n+m)$.
\end{lemma}
\begin{proof}
By definition,
  $G(m,n)=(\Sigma^n,\Sigma^{n+m},\psi'),
  G(m,n+m)=(\Sigma^{n+m},\Sigma^{n+2m},\psi''),$
where $tail'(w)=\factor{w}{1}{n}, head'(w)=\factor{w}{m+1}{m+n}$,
and
  $tail''(w)=\factor{w}{1}{m+n},
  head''(w)=\factor{w}{m+1}{2m+n}.$
So for every pair of arcs $a_1,a_2\in\Sigma^{n+m}$ of $G(m,n)$ with
$head'(a_1)=tail'(a_2)$, there is an arc
$a_1\cdot\factor{a_2}{n+1}{n+m}\in\Sigma^{n+2m}$ of $G(m,n+m)$; and
for every arc $w\in\Sigma^{n+2m}$ of $G(m,n+m)$,
  $head'(tail''(w))=\factor{w}{m+1}{m+n}=tail'(head''(w)).$
Hence, by definition, $G(m,n+m)$ is the arc-graph of $G(m,n)$.
\end{proof}

\begin{lemma}\label{lemma:equivalent}
Suppose $m\leq n$. (1) There is a $\abs{\Sigma}^n$-to-$1$ mapping
from the set of $m$-shift de~Bruijn sequences of order $n$ onto the
set of Hamilton cycles in $G(m,n)$. (2) There is a
$\abs{\Sigma}^n$-to-$1$ mapping from the set of $m$-shift de~Bruijn
sequences of order $n$ onto the set of Euler tours in $G(m,n-m)$.
\end{lemma}
\begin{proof}
Let $l=\abs{\Sigma}^n$. (1) Notice that any Hamilton cycle
$a_1,a_2,\ldots,a_l$ together with a starting arc $a_1$ uniquely
determines one $m$-shift de~Bruijn sequences of order $n$ specified
by
  \[\factor{a_1}{1}{n}\factor{a_1}{n+1}{n+m}\factor{a_2}{n+1}{n+m}\cdots\factor{a_{l-1}}{n+1}{n+m},\]
and vice versa. So the $l$-to-$1$ mapping exists. (2) Applying
Lemma~\ref{lemma:arcgraph}, this part follows from (1).
\end{proof}

\begin{theorem}\label{theorem:existence}
For any alphabet $\Sigma$, positive integers $m,n$, the $m$-shift
de~Bruijn sequences of order $n$ over $\Sigma$ exist.
\end{theorem}
\begin{proof}
First we assume $m\geq n$. Let $u_1,u_2,\ldots,u_l$ be any
permutation of the words in $\Sigma^n$ for $l=\abs{\Sigma}^n$. Then
the word $u_1{\tt0}^{m-n}u_2{\tt0}^{m-n}\cdots{\tt0}^{m-n}u_l$ is
one $m$-shift de~Bruijn sequence of order $n$ over $\Sigma$.

Now we assume $m<n$ and prove there exists an Euler tour in
$G(m,n-m)$. Then by Lemma~\ref{lemma:equivalent}, the existence of
$m$-shift de~Bruijn sequences of order $n$ over $\Sigma$ is ensured.
To show the existence of an Euler tour, we only need to verify that
$G(m,n-m)$ is connected and that $\delta^-(v)=\delta^+(v)$ for every
vertex $v$, both of which are straightforward: for every vertex $v$
in $G(m,n-m)$, $v$ is connected to the vertex ${\tt0}^{n-m}$ in both
directions and $\delta^-(v)=\delta^+(v)=\abs{\Sigma}^{m}$.
\end{proof}

\section{Counting the Number of Multi-Shift de~Bruijn Sequences}
Since $m$-shift de~Bruijn sequence of order $n$ exists, in this
section we discuss the total number of different $m$-shift de~Bruijn
sequence of order $n$, and we denote the number by $\#(m,n)$. First,
we study the degenerated case.

\begin{lemma}\label{lemma:countnlesm}
For $1\leq n\leq m$, $\#(m,n)=a^n!a^{(m-n)(a^n-1)}$, where
$a=\abs{\Sigma}$.
\end{lemma}
\begin{proof}
Let $a=\abs{\Sigma}$. By the definition of the multi-shift de~Bruijn
sequence, in the case $1\leq n\leq m$, $m$-shift de~Bruijn sequences
of order $n$ are exactly those of the form
$u_1\Sigma^{m-n}u_2\Sigma^{m-n}\cdots\Sigma^{m-n}u_l$, where $l=a^n$
and $u_1,u_2,\ldots,u_l$ is a permutation of all words in
$\Sigma^n$. Therefore, the total number of such words is
$a^n!a^{(m-n)(a^n-1)}$.
\end{proof}

To study the case $1\leq m\leq n$, we need a theorem by
van~Aardenne-Ehrenfest and
de~Bruijn~\cite{vanAardenne-Ehrenfest&deBruijn:1951}, which
describes the relation between the number of Euler tours in a
particular type of digraph and the number of Euler tours in its
arc-graph.
\begin{theorem}[van~Aardenne-Ehrenfest and de~Bruijn~\cite{vanAardenne-Ehrenfest&deBruijn:1951}]
\label{theorem:ehrenfestbruijn} Let $G=(V,A,\psi)$ be a digraph such
that $a=\delta^-(v)=\delta^+(v)$ for every $v\in V$. Then
$\abs{G^*}_E=a^{-1}(a!)^{\abs{V}(a-1)}\abs{G}_E$.
\end{theorem}

The digraph $G(m,n)$ satisfies the conditions in
Theorem~\ref{theorem:ehrenfestbruijn} with $a=\abs{\Sigma}^m$. So,
by the relation between the multi-shift de~Bruijn sequences and the
Euler tours in the word graph $G(m,n)$, we have the following
recursive expression on $\#(m,n)$.
\begin{lemma}\label{lemma:recursion}
For $m\geq1,n\geq2m$, $\#(m,n)=(a^m!)^{a^{n-m}-a^{r}}\#(m,m+r)$,
where $a=\abs{\Sigma}$, $r=n\bmod m$.
\end{lemma}
\begin{proof}
Let $a=\abs{\Sigma}$, $r=n\bmod m$. By Lemma~\ref{lemma:equivalent},
\begin{align*}
  \#(m,n)&=a^n\abs{G(m,n-m)}_E\\
    &=a^{n-m}(a^m!)^{a^{n-2m}(a^m-1)}\abs{G(m,n-2m)}_E\\
    &=(a^m!)^{a^{n-2m}(a^m-1)}\#(m,n-m)\\
    &=(a^m!)^{a^{n-2m}(a^m-1)}(a^m!)^{a^{n-3m}(a^m-1)}\#(m,n-2m)\\
    &=\ldots\\
    &=(a^m!)^{a^{n-2m}(a^m-1)}(a^m!)^{a^{n-3m}(a^m-1)}\cdots(a^m!)^{a^{r}(a^m-1)}\#(m,m+r)\\
    &=(a^m!)^{a^{n-m}-a^{r}}\#(m,m+r). \tag*{\qedhere}
\end{align*}
\end{proof}

To finish the last step of obtaining $\#(m,n)$ for $1\leq m\leq n$,
we again need two theorems, which are often used in the literature
to count the number of Euler tours in various types of digraphs.
\begin{theorem}[BEST
theorem~\cite{Tutte&Smith1941,vanAardenne-Ehrenfest&deBruijn:1951}]\label{theorem:best}
In a digraph $G=(V,A,\psi)$, $\abs{G}_E={\prod_{v\in
V}(\delta^+(v)-1)!}\,\abs{G}_A$.
\end{theorem}
\begin{theorem}[Kirchhoff's matrix tree theorem~\cite{Kirchhoff1847}]\label{theorem:kirchhoff}
In a graph $G=(V,A,\psi)$, the number of spanning trees is equal to
any cofactor of the Laplacian matrix of $G$, which is the diagonal
matrix of degrees minus the adjacency matrix.
\end{theorem}

\begin{lemma}\label{lemma:countnles2m}
For $1\leq m\leq n\leq2m$, $\#(m,n)=(a^m!)^{a^{n-m}}$, where
$a=\abs{\Sigma}$.
\end{lemma}
\begin{proof}
Let $r=n-m$ and $a=\abs{\Sigma}$. Then $0\leq r\leq m$. By
definition, $G=G(m,n-m)=(\Sigma^{r},\Sigma^{m},\psi)$. So from any
vertex to any vertex, there are $a^{m-r}$-many arcs in $G$. We
convert $G$ into a undirected graph $G'$ by omitting all self-loops;
there are $a^{m-r}$-many of them for each vertex. Since for every
pair of vertices $v_1,v_2$ there are $a^{m-r}$-many arcs joins $v_1$
to $v_2$ and correspondingly there are $a^{m-r}$-many arcs joins
$v_2$ to $v_1$, the graph $G'$ is indeed an undirected graph by our
definition. Each vertex in $G'$ is of degree $a^m-a^{m-r}$. Then the
Laplacian matrix of $G'$ is
  \[L=\left(
    \begin{array}{cccc}
      a^m-a^{m-r} & -a^{m-r} & \cdots & -a^{m-r} \\
      -a^{m-r} & a^m-a^{m-r} & \cdots & -a^{m-r} \\
      \vdots & \vdots & \ddots & \vdots \\
      -a^{m-r} & -a^{m-r} & \cdots & a^m-a^{m-r} \\
    \end{array}
  \right).\]
By Theorem~\ref{theorem:kirchhoff}, the number of arborescence
$\abs{G}_A=\abs{G'}_A$ is equal to the cofactor of $L$, which is
$(a^m)^{a^r-2}a^{m-r}=(a^m)^{a^r}/a^n$. Then by
Theorem~\ref{theorem:best}, the number of Euler tours in digraph $G$
is
$\abs{G}_E=((a^m-1)!)^{a^r}\abs{G}_A=((a^m-1)!)^{a^r}(a^m)^{a^r}/a^n=(a^m!)^{a^r}/a^n$.
Finally, by Lemma~\ref{lemma:equivalent}, the number of $m$-shift
de~Bruijn sequence of order $n$ is
$\#(m,n)=a^n\abs{G}_E=(a^m!)^{a^r}$.
\end{proof}

\begin{theorem}\label{theorem:count}
For $1\leq n\leq m$, $\#(m,n)=a^n!a^{(m-n)(a^n-1)}$, and for $1\leq
m\leq n$, $\#(m,n)=(a^m!)^{a^{n-m}}$, where $a=\abs{\Sigma}$.
\end{theorem}
\begin{proof}
For $1\leq n\leq m$, the equality $\#(m,n)=a^n!a^{(m-n)(a^n-1)}$ is
shown in Lemma~\ref{lemma:countnlesm}. Now we assume $1\leq m\leq
n$. Let $r=n\bmod m$. Then by
Lemmas~\ref{lemma:recursion},\ref{lemma:countnles2m}, we have
$\#(m,n)=(a^m!)^{a^{n-m}-a^r}\#(m,m+r)=(a^m!)^{a^{n-m}-a^r}(a^m!)^{a^r}=(a^m!)^{a^{n-m}}$.
\end{proof}

\section{Generating Multi-Shift de~Bruijn Sequences}
In this section, we study the problem of generating one $m$-shift
de~Bruijn sequence of order $n$ for arbitrary alphabet and positive
integers $m,n$. When $1\leq n\leq m$, a $m$-shift de~Bruijn sequence
of order $n$ is easy to construct as given in
Theorem~\ref{theorem:existence}. Now we consider the case $1\leq
m<n$. We will present two algorithms for generating a $m$-shift
de~Bruijn sequence of order $n$.

We claim that $m$-shift de~Bruijn sequences of order $km$ can be
generated using the ordinary de~Bruijn sequence generating
algorithm, such as described by Fredricksen~\cite{Fredricksen:1982}.
To do this, we first generate a de~Bruijn sequence $w$ of order $k$
over the alphabet $\Gamma=\Sigma^m$. Then we replace each letter of
$w$ in $\Gamma$ by the corresponding word of length $m$ over
$\Sigma$. It is easy to see that the new word is a $m$-shift
de~Bruijn sequence of order $km$.

The first algorithm of generating multi-shift de~Bruijn sequence is
to generate $m_i$-shift de~Bruijn sequences of order $k_im_i$ for
$i=1,2$ before rearranging the words to obtain an arbitrary
$m$-shift de~Bruijn sequence of order $n$. Let $1\leq m<n$ be two
integers, and $n=km+r$, where $r=n\bmod m$. The case $r=0$ is
already discussed and the case $\abs{\Sigma}=1$ is trivial. So we
assume $r\neq0$ and $\abs{\Sigma}\geq2$. We define $m_1=r$,
$n_1=(k+1)r$ and generate $w_1=\tau(m_1,n_1){\tt0}^{m_1}$ such that
$\tau(m_1,n_1)$ is a $m_1$-shift de~Bruijn sequence of order $n_1$
and $\factor{w_1}{1}{n_1}={\tt0}^{n_1}$; and define $m_2=m-r$,
$n_2=k(m-r)$ and generate $w_2=\tau(m_2,n_2){\tt0}^{m_2}$ such that
$\tau(m_2,n_2)$ is a $m_2$-shift de~Bruijn sequence of order $n_2$
and $\factor{w_2}{1}{n_2}={\tt0}^{n_2}$. Let $a=\abs{\Sigma}$,
$N_1=a^{n_1}$, $N_2=a^{n_2}$. We define
$u_i=\factor{w_1}{n_1+(i-1)m_1+1}{n_1+im_1}$,
$u_i'=u_{1+(i\bmod(N_1-1))}$,
$v_i=\factor{w_2}{n_2+(i-1)m_2+1}{n_2+im_2}$, $v_i'=v_{1+(i-1\bmod
N_2)}$. Then the following word
  \[{\tt0}^{n}\,v_1{\tt0}^{m_1}\,v_2\,\cdots\,v_{N_2-1}{\tt0}^{m_1}\,v_{N_2}u'_{(N_1-1)N_2}\,v'_1u'_1\,v'_2u'_2\,\cdots\,v'_{(N_1-1)N_2-1}u'_{(N_1-1)N_2-1}\]
is one $m$-shift de~Bruijn sequence of order $n$, where
$v_{N_2}={\tt0}^{km}$ and $u'_{(N_1-1)N_2}=u_1$.

To show the correctness, we claim that every word in
$L_1=({\tt0}^{m_1}\Sigma^{m_2})^k{\tt0}^{m_1}$ appears in
  \[w'={\tt0}^{n}\,v_1{\tt0}^{m_1}\,v_2\,\cdots\,v_{N_2-1}{\tt0}^{m_1}\]
as a factor at a modulo $m$ position exactly once. Furthermore,
since $\gcd(N_1-1,N_2)=1$, every word in
$L_2=(\Sigma^{m_1}\Sigma^{m_2})^k\Sigma^{m_1}\setminus L_1$ appears
in
 \[w''={\tt0}^{km}u_1v'_1u'_1v'_2u'_2\cdots v'_{(N_1-1)N_2-1}u'_{(N_1-1)N_2-1}\]
as a factor at a modulo $m$ position exactly once. Therefore, the
generated word is indeed a $m$-shift de~Bruijn sequence of order
$n$.

Now, we will see an example. Consider generating a $2$-shift
de~Bruijn sequence of order $5$. Then $m_1=1,n_1=3,m_2=1,n_2=2$ and
we can obtain two words $w_1=\tt00011101000$, which is
$\tau(1,3){\tt0}$, and $w_2=\tt001100$, which is $\tau(1,2){\tt0}$.
So one $2$-shift de~Bruijn sequence of order $5$ is as follows
\begin{multline*}
  \tt000001_201_200_200_2\
  \tt1_11_21_11_21_10_20_10_21_11_20_11_20_10_2
  \tt1_10_21_11_21_11_20_10_21_10_20_11_20_11_2\\
  \tt1_10_21_10_21_11_20_11_21_10_20_10_20_11_2
  \tt1_11_21_10_21_10_20_11_21_11_20_10_20_1,
\end{multline*}
where the subscripts $1$ and $2$ denote whether the letter is from
the word $w_1$ (words $u_i,u_i'$) or from the word $w_2$ (words
$v_i,v_i'$).

Now we present the second algorithm, which uses the same idea of
``prefer one'' algorithm~\cite{Martin1934} for generating ordinary
de~Bruijn sequences. Let $m,n$ be two positive integers. The
following algorithm generates a $m$-shift de~Bruijn sequence of
order $n$:
\begin{enumerate}
  \item Start the sequence $w$ with $n$ zeros;
  \item Append to the end of current sequence $w$ the lexicographically
  largest word of length $m$ such that the suffix of length $n$ of
  new sequence has not yet appeared as factor at a modulo $m$ position;
  \item Repeat the last step until no word can be added.
\end{enumerate}

To show the correctness, first we claim that when the algorithm
stops, the suffix $u$ of length $n-m$ of $w$ contains only zeros. To
see this, suppose $u$ is not ${\tt0}^{n-m}$. Since no word can be
added, all $\abs{\Sigma}^m$ words of length $n$ with prefix $u$
appear in $w$ and thus $u$ appears in $w$ as a factor at a modulo
$m$ position $\abs{\Sigma}^m+1$ times. So there are
$\abs{\Sigma}^m+1$ words of length $n$ with suffix $u$ that appear
in $w$ at a modulo $m$ position, which contradicts the definition of
the multi-shift de~Bruijn sequence. Therefore, $u={\tt0}^{n-m}$.
Furthermore, word ${\tt0}^{n-m}$ appears in $w$ as a factor at a
modulo $m$ position $\abs{\Sigma}^m+1$ times and thus all words in
$\Sigma^m{\tt0}^{n-m}$ appear in $w$ as a factor at a modulo $m$
position. By the algorithm, no word of length $n$ can appear twice
in $w$ at a modulo position. So, in order to prove the correctness
of the algorithm, it remains to show every word of length $n$
appears in $w$ as a factor at a modulo $m$ position. Suppose a word
$v$ does not appear in $w$ at a modulo $m$ position. Then
$\factor{v}{m+1}{n}\neq{\tt0}^{n-m}$ and the word
$\factor{v}{m+1}{n}{\tt0}^m$ does not appear in $w$ as a factor at a
modulo $m$ position as well; otherwise, there are $\abs{\Sigma}^{m}$
appearance of $\factor{v}{m+1}{n}$ in $w$ at a modulo $m$ position,
which means $v$ appears in $w$ as a factor at a modulo $m$ position.
Repeat this procedure, none of the words
$\factor{v}{m+1}{n}{\tt0}^m$, $\factor{v}{2m+1}{n}{\tt0}^{2m}$,
$\ldots$, $\factor{v}{\lfloor n/m\rfloor m+1}{n}{\tt0}^{\lfloor
n/m\rfloor m}$ appears in $w$ as a factor at a modulo $m$ position.
But for $\lfloor n/m\rfloor m\geq n-m$, we proved that
$\factor{v}{\lfloor n/m\rfloor m+1}{n}{\tt0}^{\lfloor n/m\rfloor m}$
appears in $w$ as a factor at a modulo $m$ position, a
contradiction. Therefore, every word of length $n$ appears at a
modulo $m$ position.

Now, we use the algorithm to generate one $2$-shift de~Bruijn
sequence of order $5$. Starting from $\tt00000$, since $\tt00011$
does not appear as a factor at a modulo $2$ position, we append
$\tt11$ to the current sequence $\tt00000$. Repeating this procedure
and appending words $\tt11$, $\tt11$, $\tt10$, $\tt11$, \ldots,
finally we obtain the word:
  \[\tt0000011111 1101110101 1011101100 1110011001 
  \tt0100110001 0000101010 0010000\]

If we circularly move the prefix ${\tt0}^n$ to the end, the sequence
generated by the second algorithm is the lexicographically largest
$m$-shift de~Bruijn sequence of order $n$.

\section{Application in the Frobenius Problem in a Free Monoid}
The study of multi-shift de~Bruijn sequences is inspired by a
problems of words, called the Frobenius problem in a free monoid.
Given $k$ integers $x_1,\ldots,x_k$, such that
$\gcd(x_1,\ldots,x_k)=1$, then there are only finitely many positive
integers that \emph{cannot} be written as a non-negative integer
linear combination of $x_1,\ldots,x_k$. The integer \emph{Frobenius
problem} is to find the largest such integer, which is denoted by
$g(x_1,\ldots,x_k)$. For example, $g(3,5)=7$.

If words $x_1,\ldots,x_k$, instead of integers, are given such that
there are only finitely many words that \emph{cannot} be written as
concatenation of words from the set $\set{x_1,\ldots,x_k}$, the
\emph{Frobenius problem in a free monoid}~\cite{STACS2008} is to
find the longest such words. If all $x_1,\ldots,x_k$ are of length
either $m$ or $n$, $0<m<n$, there is an upper bound: the length of
the longest word that cannot be written as concatenation of words
from the set $\set{x_1,\ldots,x_k}$ is less than or equal to
$g(m,l)=ml-m-l$, where $l=m\Sigma^{n-m}+n-m$.~\cite{STACS2008}
Furthermore, the upper bound is tight and the construction is based
on the multi-shift de~Bruijn sequences. We denote the set of all
words that can be written as concatenation of words in $S$,
including the empty word, by $S^*$.

\begin{theorem}~\cite{STACS2008}
There exists $S\subseteq\Sigma^m\cup\Sigma^n$, $0<m<n$, such that
$\Sigma^*\setminus S^*$ is finite and the longest words in
$\Sigma^*\setminus S^*$ constitute exactly the language
$(\tau\Sigma^m)^{m-2}\tau$, where $\tau$ is a $m$-shift de~Bruijn
sequence of order $n-m$.
\end{theorem}

For example, for any set of words $S\subseteq
U=\set{{\tt0,\tt1}}^3\cup\set{{\tt0,\tt1}}^7$ such that
$\set{{\tt0,\tt1}}^*\setminus S^*$ is finite, the longest words in
$\set{{\tt0,\tt1}}^*\setminus S^*$ are of length less than or equal
to $g(3,3\cdot2^4+4)=g(3,52)=101$. To construct $S$ to reach the
upper bound, we first choose an anbitrary $3$-shift de~Bruijn
sequence of order $4$ as $\tau=\tt000 011 111 111 011 010 110 110
010 001 101 101 001 000 1001000$. Then based on $\tau$, we construct
the set $S=U\setminus\{$ $\tt0000111$, $\tt0111111$, $\tt1111110$,
$\tt1110110$, $\tt0110101$, $\tt0101101$, $\tt1101100$,
$\tt1100100$, $\tt0100011$, $\tt0011011$, $\tt1011010$,
$\tt1010010$, $\tt0010001$, $\tt0001001$, $\tt1001000$ $\}$. We have
$L=\set{{\tt0,\tt1}}^*\setminus S^*=\tau\set{{\tt0,\tt1}}^2\tau$ and
one of the longest words in $L$ of length exactly $101$ is given
below:
\begin{multline*}
\tt 0000111111 1101101011 0110010001 1011010010
0010010001\\
\tt 110000111111110110101101100100 011011010010001001000.
\end{multline*}


\section{Conclusion}
In this paper, we generalized the classic de~Bruijn sequence to a
new multi-shift setting. A word $w$ is a $m$-shift de~Bruijn
sequence $\tau(m,n)$ of order $n$, if each word of length $n$
appears exactly once as a factor at a modulo $m$ position. An
ordinary de~Bruijn sequence is a $1$-shift de~Bruijn sequence.

We showed the total number of distinct $m$-shift de~Bruijn sequences
of order $n$ is $\#(m,n)=(a^n)!a^{(m-n)(a^n-1)}$ for $1\leq n\leq m$
and is $\#(m,n)=(a^m!)^{a^{n-m}}$ for $1\leq m\leq n$, where
$a=\abs{\Sigma}$. This result generalizes the formula
$(a!)^{a^{n-1}}$ for the number of ordinary de~Bruijn
sequences~\cite{vanAardenne-Ehrenfest&deBruijn:1951}. Here we use an
ordinary word form; if counting the sequences in a circular form,
then the number is to be divided by $a^n$.

We provided two algorithms for generating a $m$-shift de~Bruijn
sequence of order $n$. The first algorithm is to rearrange factors
from two simpler multi-shift de~Bruijn sequences, where the order is
a multiple of the shift. The second is the analogue of the ``prefer
one'' algorithm (for example, see~\cite{{Fredricksen:1982}}) for
generating ordinary de~Bruijn sequence.

The multi-shift de~Bruijn sequence has application in the Frobenius
problem in a free monoid by providing constructions of examples. It
will be interesting to see that this generalized concept of the
de~Bruijn sequence can help in other fields of theoretical computer
science and discrete mathematics.

\ack{The author would like to thank Prof. Jeffrey Shallit for
valuable discussion.}


\begin{thebibliography}{10}

\bibitem{vanAardenne-Ehrenfest&deBruijn:1951}
{T. van} {Aardenne-Ehrenfest} and N.~G. de~Bruijn.
\newblock Circuits and trees in oriented linear graphs.
\newblock {\em Simon Stevin}, 28:203--217, 1951.

\bibitem{deBruijn:1946}
{N. G. de} Bruijn.
\newblock A combinatorial problem.
\newblock {\em Indag. Math.}, 8(4):461--467, 1946.

\bibitem{FlyeSainteMarie:1894}
C.~{Flye Sainte-Marie}.
\newblock Solution to question nr. 48.
\newblock {\em L'Interm\'ediaire Math.}, 1:107--110, 1894.

\bibitem{Fredricksen:1982}
H.~Fredricksen.
\newblock A survey of full length nonlinear shift register cycle algorithms.
\newblock {\em SIAM Review}, 24(2):195--221, 1982.

\bibitem{Fredricksen&Kessler1977}
H.~Fredricksen and I.~J. Kessler.
\newblock Lexicographic compositions and de~bruijn sequences.
\newblock {\em J. Combin. Theory Ser. A}, 22:17--30, 1977.

\bibitem{Golomb:1967}
S.~W. Golomb.
\newblock {\em Shift Register Sequences}.
\newblock Holden-Day, 1967.

\bibitem{Good:1946}
I.~J. Good.
\newblock Normal recurring decimals.
\newblock {\em J. London Math. Soc.}, 21(3):167--169, 1946.

\bibitem{STACS2008}
{J.-Y.} Kao, J.~Shallit, and Z.~Xu.
\newblock The {Frobenius} problem in a free monoid.
\newblock In {\em STACS 2008, Proc. 25th Internat. Symp. Theoretical Aspects of
  Comp. Sci.}, pages 421--432, 2008.

\bibitem{Kirchhoff1847}
G.~Kirchhoff.
\newblock {\"U}ber die {Aufl{\"o}sung} der {Gleichungen}, auf welche man bei
  der untersuchung der linearen verteilung galvanischer {Str\"ome} gef{\"u}hrt
  wird.
\newblock {\em Ann. Phys. Chem.}, 72:497--508, 1847.

\bibitem{Knuth1969}
D.~E. Knuth.
\newblock {\em The Art of Computer Programming}.
\newblock Addison-Wesley, 1969.

\bibitem{Martin1934}
M.~H. Martin.
\newblock A problem in arrangements.
\newblock {\em Bull. Amer. Math. Soc.}, 40:859--864, 1934.

\bibitem{Ralston:1982}
A.~Ralston.
\newblock De~{Bruijn} sequences --- a model example of the interaction of
  discrete mathematics and computer science.
\newblock {\em Math. Mag.}, 55(3):131--143, 1982.

\bibitem{Rees:1946}
D.~Rees.
\newblock Note on a paper by {I. J. Good}.
\newblock {\em J. London Math. Soc.}, 21(3):169--172, 1946.

\bibitem{Tutte&Smith1941}
W.~T. Tutte and C.~A.~B. Smith.
\newblock On unicursal paths in a network of degree 4.
\newblock {\em Amer. Math. Monthly}, 48(4):233--237, 1941.

\end{thebibliography}

\end{document}